\documentclass{article}

\usepackage[english]{babel}

\usepackage[letterpaper,top=2cm,bottom=2cm,left=3cm,right=3cm,marginparwidth=1.75cm]{geometry}

\usepackage{amsfonts}
\usepackage{amsmath}
\usepackage{graphicx}
\usepackage[colorlinks=true, allcolors=blue]{hyperref}
\usepackage[english]{babel}
\usepackage{amsthm}
\usepackage[linesnumbered,ruled,lined,noend]{algorithm2e}
\SetKw{KwBy}{by}  

\newtheorem{definition}{Definition}
\newtheorem{theorem}{Theorem}
\newtheorem{lemma}{Lemma}
\newtheorem{observation}{Observation}
\newtheorem{remark}{Remark}

\title{Shapley Meets Tutte}
\author{Martin Loebl\\
Agate Centre\\ Charles University Prague\\
loebl@kam.mff.cuni.cz}
\begin{document}
\maketitle

\begin{abstract}
We initiate the study of cooperative games where there exist {\em groups of pre-aligned agents}. 
 For example, in a road network, the agents are cross-roads and pre-aligned groups are road-segments. For a set of datasets, the agents are attributes and each database which connects two (or more) attributes is pre-aligned.
In this paper, each pre-aligned group has size two.

 Our goal is to find a way to determine the contribution of each individual local connection (pre-aligned couple) to the connectivity of the network, having an adversarial attack in mind or planning a defense of
the network, or wanting to split the profits of the network between various owners of the individual local connections of the network. We model these contributions by Shapley values of the connectivity augmented 'local values' which are characteristic functions determined, e.g., by failure probabilities of local connections.

We link the concepts of potential and Shapley values of these connectivity augmented local values to the world of chromatic and Tutte polynomials and the partition function of the Potts model from statistical physics.
\end{abstract}

\newpage
\section{Introduction}
\label{sec.intro}
Consider a connection network, e.g., a road network, a water pipeline network, an electric network, a network of optical cables, or a network of bus routes or railroads. Assume parts of the network are owned by different companies. Clearly, it is advantageous for customers to use the connection network regardless of which part is owned by which company. Can this be achieved, i.e., is there a fair way of sharing profits?
Similarly, consider a set of data sets that connect information. Assume that several companies own parts of the set. Is there a fair way to share profits from customers who use subsets of the set of databases to connect information that they are interested in? 
One can imagine a special case of this situation where databases are individual people and each person speaks two (or more) languages. 
Similarly, is there a way to determine the values of individual local connections in a network {\em for the network's connectivity}, having an adversarial attack in mind or planning a defense of the network? What if, in addition, the local connections may fail, with some probability.

All activities considered live in the {\em couple graph} $G=(V,E)$, where $V$ denotes the set of its vertices and $E$ denotes its set of edges. This graph models the connection network; the pre-aligned couples, i.e., individual connections (pipes, road segments, databases) are modeled by the edges of this graph, while the vertices model the items that are connected, e.g., crossroads or attributes or languages. 

The situation (business) we want to model and study is connectivity. We want to study how each edge, i.e., pre-aligned couple contributes to connectivity of the couple graph.

Cooperative Game Theory analyzes how groups of players (coalitions) can  fairly divide the
joint payouts. Hence, our goal is to model the connectivity of the couple graph as a cooperative game in which the players are the edges of the couple graph.

A cooperative game is usually determined by a characteristic function that gives value to each coalition. For such cooperative games, the most influential solution concept for modeling individual contributions of players to the total value of the game is their {\em Shapley value}. Let $\mathcal Q^+$ denote the set of non-negative rational numbers.

\medskip\noindent

{\em Our objective is to model the connectivity of the couple graph $G$ by constructing a value giving characteristic function
$w(G): 2^E\rightarrow \mathcal Q^+$ and studying the Shapley values of the edges.} 

\medskip\noindent

In situations we want to study, the value of a coalition has two components: it consists of the {\em local value} and the {\em connectivity value}. 

The local value of each subset of edges is given by its characteristic function  $w: 2^E\rightarrow \mathcal Q^+$ such that $w(\emptyset)= 0$ and is usually well known. We concentrate on local characteristic functions where for each $\emptyset\neq A\subset E$, $w(A)= \prod_{e\in A}w(e)$; for example, if $w(e)$ is the probability that edge $e$ does not fail, then $w(A)$ is equal to the {\em reliability} of coalition (set of edges) $A$. 

It is also natural to consider situations where an alignment in a given couple is {\em negative}. This can be included in our model considering the {\em local values} given by characteristic functions $w: 2^E\rightarrow \mathcal Q^+$ such that $w(A)= \prod_{e\in A}w(e)$ and letting $w(e)$ very small for negatively aligned couples.

Aside of the local value, each subset of edges (coalition) has yet another value, important for the operation of the connection network, which represents the contribution of the subset to {\em connecting the network}. 
This {\em connectivity value} of $A\subset E$ depends on the set of connectivity components of the graph $(V,A)$, and on the size of $A$. 

Hence, the construction of the characteristic function $w(G)$ introduced in this paper involves two types of value of coalitions (subsets of edges). 

\subsection{Acknowledgement} This project receives funding from the Horizon EU Framework Programme under Grant Agreement No. 101183743 (project AGATE). I would like to thank Martin Černý, Anna de Mier, Andrew Goodall and Ole Jann for helpful discussions.

\section{Basic Setting and Main Contribution}  
\label{sub.mmain}
We initiate the study of situations where there exist {\em groups of pre-aligned agents}. Given a partition of agents, a pre-aligned group contributes to its {\em frustration} if it is split among different classes of the partition.
In this paper, each pre-aligned group has size two, the general setting being a work in progress. For example, in a road network, the agents are cross-roads and pre-aligned groups are road-segments. For a set of datasets, the agents are attributes and each database which connects two (or more) attributes is pre-aligned.

 The structure of pre-aligned couples is determined by the graph $G= (V,E)$, where $V$ is the set of vertex-agents, and the set of pre-aligned couples forms the set $E$ of edges. We call this graph {\em the couple graph} of the game.
We first specify the class of cooperative games with pre-aligned couples we are interested in.

\begin{definition}
    \label{def.eda}
Let $G= (V,E)$ be a graph. We let 
$\mathcal E= \{w: 2^E\rightarrow \mathcal Q^+ \text{ such that } w(\emptyset) = 0\}$
 be the class  of cooperative games in the characteristic function form, where the edges of $G$ are the players.
\end{definition}

 \subsection{Our Contribution}
 \label{sub.g}
\begin{itemize}
    \item 
    Our main goal is to augment 'local' values of subsets of edges (pre-aligned couples), by including their 'connectivity' values, i.e., to favor subsets with smaller number of connectivity components and subsets of bigger size. For the resulting games, we study the Shapley values of edges. 
    \item
    We link the potential and Shapley values of these connectivity augmented local values to the world of chromatic and Tutte polynomials and the partition function of the Potts model from statistical physics.
\end{itemize}

\subsection{Main results} 
\label{sub.mres}

\begin{itemize}
    \item
    In the next section, we show how to modify, in a consistent way, the local values of subsets of edge-agents in order to stress the importance of connectivity. The resulting cooperative games are called {\em connectivity augmented games}.
    \item 
In Section \ref{s.pot} we study the potential of connectivity augmented games. We present our first main result in Theorem \ref{thm.chrom} which describes the potential of each connectivity augmented game using the chromatic polynomial of contractions of the couple graph.   
    \item 
In Section \ref{s.tutte}, we consider a class of 
{\em couple games} which are connectivity augmented games with local value function $w:E\rightarrow \mathcal Q^+$, $w(\emptyset)= 0$ such that for each $\emptyset\neq A\subset E$,  $w(A)= \prod_{e\in A}w(e)$.
Such characteristic functions model, for example, for each subset of edges $A$ the reliability of $A$, i.e., probability that $A$ does not fail.

 Our main results are: 
 
 (1) Theorem \ref{thm.bad} stating that the potential of each couple game is equivalent to the difference of the multivariate bad coloring polynomial (bad coloring polynomial is equivalent to the Tutte polynomial) and the chromatic polynomial of the couple graph.
 
 (2) Theorem \ref{thm.dcc} and Theorem \ref{thm.dcc1} stating that the Shapley value of each edge = pre-aligned couple $a$ can also be expressed by the chromatic polynomial and the
 multivariable bad coloring polynomial of the couple graph with $a$ contracted. 
 \item
 In Section \ref{s.qchrom}, we present an extension of the basic setting where we exchange the chromatic polynomial with the q-chromatic function.
    \item 
    This work thus connects basic notions of economics, mathematics and statistical physics: Shapley value, chromatic polynomial, Tutte polynomial, and the partition function of the Potts model.
\end{itemize}

\section{Connectivity augmented games}
\label{s.connectivity}
Let $G= (V,E)$ be a graph. 
Assume that we want to modify the {\em values of the subsets of edges} in the games of $\mathcal E$,  to favor the subsets with fewer connectivity components and of bigger size. Clearly, we want to do the modification in a consistent way: in which sense?

\subsection{Space of cooperative games}
\label{sub.c}

Each cooperative game of $\mathcal E$, determined by its characteristic function $w:2^E\rightarrow \mathcal Q$, $w(\emptyset)= 0$,
can be viewed as a rational vector indexed by subsets of $E$ and is clearly a linear combination of the standard basis of $(\mathcal Q)^{2^E}$, where the coefficient of a subset $A\subset E$ is $w(A)$. However, for cooperative games, it is customary to use a basis other than the standard one.

It is well known (see, e.g., \cite{S}) that each cooperative game of $\mathcal E$, determined by its characteristic function $w:2^N\rightarrow \mathcal Q^+$, can be written as a linear combination of {\em basic games}. 

\begin{definition}
    \label{def.basic}
The basic game $d(R), \emptyset\neq R\subset E$ is defined as follows: if $R\subset S\subset E$ then $d(R)(S)= 1$, otherwise $d(R)(S)= 0$. 
\end{definition}

For notational reasons, we let $d(\emptyset)(S)= 0$ for each $S\subset E$. The explicit form of the coefficients is written in the next lemma (see \cite{S}).

\begin{lemma}
    \label{l.coef}
Each cooperative game of $\mathcal E$ with characteristic function $w$ can be expressed as
$$
w=\sum_{\emptyset\neq R\subset E}c_R(w) d(R)
$$ 
where $c_R(w)= \sum_{T\subset R} (-1)^{|R|-|T|} w(T)$. 
\end{lemma}

\begin{remark}
\label{r.syn}
Lemma \ref{l.coef} is often interpreted as describing cooperative games using {\em synergies}. Synergy is defined as a function 
$w': 2^E\rightarrow \mathcal Q$ that satisfies, for each $R\subset E$, that
$w(R)=\sum _{T\subset R}w'(T)$.
Using the principle of inclusion and exclusion, we get 
$w'(R)=\sum _{T\subset R}(-1)^{|R|-|T|}w(T)= c_R(w)$.
\end{remark}

\begin{definition}
    \label{def.cons}
    We say that a bijection $f: \mathcal E\rightarrow \mathcal E$ is {\em consistent}, if (1) for each $\emptyset\neq R\subset E$ there is a rational number $n_R$ such that $f(d(R))= n_Rd(R)$, and (2) for each cooperative game of $\mathcal E$ with characteristic function $w$, $f(w)= \sum_{\emptyset\neq R\subset E}c_R(w) n_Rd(R)$, where the coefficients $c_R(w)$ are introduced in Lemma \ref{l.coef}. 
\end{definition}

 \subsection{Augmenting preference of connectivity}
 \label{sub.pref}
 
 The construction by which we accomplish the desired consistent value augmentation uses the setting described above. We proceed as follows.

\begin{itemize}
\item
It is quite clear how to do the connectivity augmentation of values for {\em basic games} of Definition \ref{def.basic}: We linearly increase the value of a subset of edges (coalition) with the number of its edges and decrease the value exponentially with the number of its connectivity components. These {\em basic connectivity augmented games} are explicitly defined in Definition \ref{def.graph-gamebasis}.

\item 
In order to do connectivity augmentation consistently for all games of $\mathcal E$, we use Lemma \ref{l.coef}.

The change of each basic game into the {\em basic connectivity augmented game} in Definition \ref{def.graph-gamebasis}  
modifies each cooperative game $w$ of $\mathcal E$, using the linear combination with the same coefficients $(c_R(w))_{\emptyset\neq R\subset E}$, of {\em basic connectivity augmented games instead of the basic games}. 

We do not know these modified games explicitly in the standard basis, but we know them explicitly in the basis formed by the basic connectivity augmented games.
\end{itemize}

\begin{definition}
    \label{def.comp}
    For $A\subset E$ let $c(A)$ denote the number of connectivity components of $(V,A)$. We say that $A\subset E$ is {\em flat} if for each $e\in E\setminus A$, $c(A)> c(A\cup\{e\})$. 
    \end{definition}

\begin{definition}[Basic connectivity augmented games]    \label{def.graph-gamebasis}
    Let $G=(V,E)$ be a graph, and $x$ be a number. For $\emptyset\neq R\subset E$, let $b(R,G,x)$ be the characteristic function (of a cooperative game) defined as follows:
if $R\subset S\subset E$ then $b(R,G,x)(S)= |R|x^{c(R)}$, otherwise $b(R,G,x)(S)= 0$.
We call $b(R,G,x)$ a {\em basic connectivity augmented game}.
\end{definition}
    
\begin{definition}[Connectivity augmented games]
    \label{def.conaugame}
 Let $G=(V,E)$ be a graph. Let $w:2^E\rightarrow \mathcal Q$ be a cooperative game of $\mathcal E$. We denote by $w(G)$ the cooperative game $w(G)=\sum_{\emptyset\neq R\subset E}c_R(w) b(R,G,x)$. We call $w(G)$ {connectivity augmented $w$}, or a {\em connectivity augmented game}.
\end{definition}

\subsection{Back to the applications}
\label{sub.interpret}
As written in the Introduction, this paper aims to suggest a fair way to split the profits generated by a given connection network among the owners of parts of the network. We first model the network as a game of $\mathcal E$ whose local characteristic function, usually denoted by $w$, is explicitly known, e.g., for each $\emptyset\neq A\subset E$, we have $w(A)= \prod_{e\in A} w(e)$. 

Then, in this section (Section \ref{s.connectivity}), the local characteristic function $w$ is changed by connectivity augmentation. This modification models the {\em connectivity value} of a subset of pre-aligned couples, anticipated in the Introduction. 

The characteristic function of the resulting couple game  can be calculated in principle using Definition \ref{def.conaugame}, but the calculation is not clearly efficient. We neither know how to efficiently find the {\em total value} of a couple game.

This total value is clearly needed to determine the profit that we want to distribute. Why bother trying to find a way to distribute something that we cannot calculate?

Here is why: the {\em total value} of the resulting couple game models the total profit. We do not know how to calculate the total value, but in each practical case, we know the total profit! 

If we explicitly determine how to distribute the (implicit) total value, we can use the resulting distribution as {\em individual demands} and split the actual total profit using, e.g., contested garment rule (see \cite{AM}).

In the rest of the paper, we show that the potential and the Shapley values of a couple game, which distribute the implicit total value, correspond to the evaluations of the Potts partition function (multivariable Tutte polynomial) of contractions of the couple graph of the game.

\section{Shapley values and potential of a connectivity augmented game}
\label{s.pot}
In this section, we present a new formula for Shapley values in connectivity augmented games.

\subsection{Shapley values in a cooperative game}
\label{sub.shap}
The Shapley value, introduced in 1953 by Lloyd Shapley (see \cite{S}), is a solution concept to fairly distribute the total gains or costs among a group of players who have collaborated. It is used extensively in many areas of game theory, economics, and, e.g., machine learning.

It can be defined in many equivalent ways. We start with a definition below since it is useful for our purposes. Without loss of generality, we assume that the set of agents is the edge-set of a graph.

\begin{definition}[Shapley]
 \label{def.shapleygraphic}   
Let $G=(V,E)$ be a graph. The Shapley value of the agent $e\in E$ in the cooperative game $w$ of Lemma \ref{l.coef}, denoted by $\phi_e(w)$, is the unique linear map satisfying, for each $\emptyset\neq R\subset E$:
$\phi_e(d(R))= 1/|R|$ if $e\in R$ and $\phi_e(d(R))= 0$ if $e\notin R$.
Hence,
$$
\phi_e(w)= \sum_{e\in R\subset E}1/|R|c_R(w).
$$
\end{definition}

The Shapley value is the only solution that satisfies four fundamental properties: efficiency, symmetry, additivity, and the null player property, which are widely accepted as defining a fair distribution.
\begin{itemize}
    \item Efficiency: $\sum_{e\in E}\phi_e(w)= w(E)$,
    \item Symmetry: if two agents marginally contribute  equally to each coalition, then their Shapley values are equal,
    \item Additivity: it is linear: this property has been used in Definition \ref{def.shapleygraphic},
    \item Null player: Shapley value of an agent with all marginal contributions zero is equal to zero.
\end{itemize}

Another widely used definition is that the Shapley value of $e$ is the average, over all linear orders $L$ of $E$, of the marginal contribution of $e$ to the set of all predecessors of $e$ in $L$.

\subsection{Shapley values in connectivity augmented games}
\label{sub,sha}

In connectivity augmented games, the Shapley values are calculated by the next lemma.

\begin{lemma}
 \label{l.shapleygraphic}   
Let $G=(V,E)$ be a graph. Let $w:2^E\rightarrow \mathcal Q^+$ be a cooperative game of $\mathcal E$. The Shapley value of the edge $e\in E$ in the connectivity augmented game $w(G)$, denoted by $\phi_e(w(G))$, is the unique linear map that satisfies (1) for each $\emptyset\neq R\subset E$:
$\phi_e(b(R,G,x))= x^{c(R)}$ if $e\in R$ and $\phi_e(b(R,G,x))= 0$ if $e\notin R$.
Hence (2) for each $e\in E$,
$$
\phi_e(w(G))= \sum_{e\in R\subset E}x^{c(R)}c_R(w).
$$
\end{lemma}
\begin{proof}
The first part follows from Definition \ref{def.graph-gamebasis} and from the linearity of the Shapley value, since $b(R,G,x)= |R|x^{c(R)}d(R)$, and $\phi_e(d(R))= 1/|R|$ by Definition \ref{def.shapleygraphic}.

The second part follows from the linearity of the Shapley value and Definition \ref{def.conaugame}.

\end{proof}

\subsection{Potential of a cooperative game}
\label{sub.pot}

The Potential of a cooperative game was introduced in \cite{HM}.
Without loss of generality, we again assume that the set of agents is the edge-set of a graph.

\begin{definition}[Hart, Mas-Colell]
    \label{def.pot}
Let $G=(V,E)$ be a graph. The {\em Potential} $P(w)$ of a cooperative game $w$ of $\mathcal E$ is defined by the requirement that 
the marginal contribution of agent $e$ to $P(w)$, i.e., $P(w)- P(w|({E\setminus \{e\}})$, is equal to the Shapley value $\phi_e(w)$.
\end{definition}

\begin{lemma}
 \label{l.potgraphic}  
 Let $G=(V,E)$ be a graph. Let $w:2^E\rightarrow \mathcal Q^+$ be a cooperative game of $\mathcal E$. The potential $P(w(G))$ of the connectivity augmented game $w(G)$ satisfies:
$$
P(w(G))= \sum_{\emptyset\neq R\subset E}x^{c(R)}c_R(w)=
$$
$$
\sum_{R\subset E}x^{c(R)}\sum_{T\subset R} (-1)^{|R|-|T|} w(T)=
$$
$$
\sum_{T\subset E}w(T) \sum_{T\subset R}(-1)^{|R|-|T|}x^{c(R)}.
$$
\end{lemma}
\begin{proof}
The first equation follows from the definition of the Potential (Definition \ref{def.pot}) and Lemma \ref{l.shapleygraphic}.
The second equation follows from Lemma \ref{l.coef} and the assumption that $w(\emptyset)= 0$, and the last equation is clear.

\end{proof}

We conclude this section by a formula connecting the Potential of a connectivity augmented game on graph $G$ with the chromatic polynomial of contractions of $G$. 

\subsection{Chromatic polynomial and potential of connectivity augmenting game}
\label{sub.chrompot}

The {\em chromatic polynomial} of a graph $G$, denoted by $C(G,x)$, is a polynomial which encodes the number of distinct ways to properly color the vertices of G (where colorings are counted as distinct even if they differ only by a permutation of colors). A {\em proper coloring} of vertices assigns different colors to the two vertices of any edge.

The chromatic polynomial was introduced by G.D. Birkhoff in 1912 (see \cite{BL}), who proved the following theorem.

\begin{theorem}[Birkhoff]
    \label{thm.bwh}
 Let $G=(V,E)$ be a graph, $x$ be a natural number, and  $C(G,x)$ denote the number of proper colorings of $G$ by at most $x$ colors. Then
 $$
 C(G,x)= \sum_{A\subset E}(-1)^{|A|}x^{c(A)},
 $$
 where $c(A)$ denotes the number of connectivity components of graph $(V,A)$ (see Definition \ref{def.comp}).
\end{theorem}

\begin{definition}[chromatic polynomial]
\label{def.ch}
Let $G=(V,E)$ be a graph and let $x$ be a variable. The chromatic polynomial of $G$ is 
 $$
 C(G,x)= \sum_{A\subset E}(-1)^{|A|}x^{c(A)}.
 $$
\end{definition}

Before connecting the Potential of a connectivity augmented game and the chromatic polynomial, we need one more definition, of the {\em contraction} of a subset of edges.

\begin{definition}
  \label{def.contr} 
  Let $G=(V,E)$ be a graph and let $A\subset E$. We denote by 
  $G.A$ the graph obtained from $G$ by {\em contraction} of $A$, i.e., $G.A= (W,F)$ where $W$ is the set of all connectivity components of graph $(V,A)$ and $F$ has exactly one edge $e'$ for each edge $e$ of $E\setminus A$. Such edge $e'$ connects vertices $u,v$ of $W$ representing subsets of vertices $S_u, S_v$ of $G$ such that in $G$, $e$ connects $S_u, S_v$. Hence, $G.A$ may have loops or multiple edges.
\end{definition}

\begin{theorem}
\label{thm.chrom}
Let $G=(V,E)$ be a graph and let $w$ be a cooperative game of $\mathcal E$. The Potential $P(w(G))$ of the connectivity augmented game $w(G)$ satisfies:

$$
P(w(G))= \sum_{R\subset E \text{ flat}} w(R)C(G.R,x),
$$
where $G.R$ is the graph obtained from $G$ by contraction of edges of $R$, and $C(H,x)$ denotes the chromatic polynomial of graph $H$.
\end{theorem}

\begin{proof}
By Lemma \ref{l.potgraphic},
$$
P(w(G))= \sum_{T\subset E}w(T) \sum_{T\subset R}(-1)^{|R|-|T|}x^{c(R)}=
$$
   $$
\sum_{T\subset E}w(T) \sum_{Q\subset E\setminus T}(-1)^{|Q|}x^{c((V, Q\cup T)}.
$$ 
Now, we notice that, for $T\subset E$ not flat,
$$
 \sum_{Q\subset E\setminus T}(-1)^{|Q|}x^{c((V, Q\cup T)}= 0.
$$ 
Hence,
 $$
\sum_{T\subset E}w(T) \sum_{Q\subset E\setminus T}(-1)^{|Q|}x^{c((V, Q\cup T)}=
$$
 $$
\sum_{T\subset E \text{ flat}}w(T) \sum_{Q\subset E\setminus T}(-1)^{|Q|}x^{c((V, Q\cup T)}=
$$
 $$
\sum_{T\subset E \text{ flat}}w(T) \sum_{Q\subset E\setminus T}(-1)^{|Q|}x^{c(V(G.T),Q)}=
$$ 
$$
\sum_{T\subset E \text{ flat}}w(T)C(G.T,x).
$$ 
\end{proof}

We end this section by defining a polynomial equivalent  to the chromatic polynomial (see \cite{B}).

\begin{definition}[characteristic polynomial]
\label{def.chp}
Let $G=(V,E)$ be a graph and let $x$ be a variable. For $A\subset E$ let $r(A)= |V|-c(A)$. The characteristic polynomial of $G$ is 
 $$
 p(G,x)= \sum_{A\subset E}(-1)^{|A|}x^{r(E)- r(A)}.
 $$
\end{definition}
 The following observation follows directly from the definition.

\begin{observation}
    \label{o.pchcol}
    $x^{c(E)}p(G,x)= C(G,x)$
\end{observation}

\section{Shapley values of couple game and Tutte polynomial}
\label{s.tutte}

In this section, we study the main object of this paper, the {\em couple game} (see Definition \ref{def.coup}). We introduce the Tutte polynomial, a seminal notion of discrete mathematics, and its multivariable version which originates from the partition function of the Potts model. Finally, we prove our main results.

\begin{definition}
    \label{def.coup}
 The connectivity augmented game $w(G)$ is called {\em couple game} if the local value function $w:E\rightarrow \mathcal Q^+$ satisfies for each $\emptyset\neq A\subset E$, $w(A)= \prod_{e\in A}w(e)$.
\end{definition}

\subsection{Couple games and Reliability}
\label{sub.rel}
There are several natural couple games; in particular, if $w(e)$ is interpreted as the probability that the edge $e$ does not fail, then $w(A)= \prod_{e\in A}w(e)$ is equal to the {\em reliability} of the coalition (set of edges) $\emptyset\neq A$. 
The Shapley value $\phi_e(w(G))$ of such a couple game is interpreted as the value of edge $e$ for the connectivity of a network with given probabilities of local failures.

\subsection{Tutte polynomial}
\label{sub.chrom}
We introduce a seminal concept of discrete mathematics, the Tutte polynomial, defined by W.T. Tutte in 1947 (see \cite{T}) building upon the rank generating function of H. Whitney (see \cite{W}).

\begin{definition}[Tutte polynomial]
    \label{def.tutte}
    Let $G= (V,E)$ be a graph. The {\em Tutte polynomial} of $G$, denoted by $T(G,x,y)$, is defined by the following formula:
    $$
    T(G,x,y)= \sum_{A\subset E}(x-1)^{z(A)}(y-1)^{n(A)},
    $$
    where $z(A)= r(E)- r(A)$ and $n(A)= |A|- r(A)$; for $A\subset E$, $r(A)= |V|- c(A)$. 
\end{definition}

There are many different equivalent definitions of the Tutte polynomial. Here we mention only three, of which the first two are well known and the third one is less known, nevertheless we need it. 
The first one states a straightforward observation that the Tutte polynomial can be expressed as a 2-variable extension of the chromatic polynomial.

\begin{theorem}
 \label{thm.tch}
 The Tutte polynomial is equivalent to 
 $$
 W(G,x,y)= \sum_{A\subset E}y^{|A|}x^{c(A)}.
 $$
\end{theorem}

The second is called the {\em contraction deletion formula}.

\begin{theorem}[deletion contraction formula]
\label{thm.delcontr}
Let $G=(V,E)$ be a graph and $e\in E$.
\begin{itemize}
    \item 
    If $e$ is a loop then $T(G,x,y)= yT(G-e,x,y)$,
    \item 
    If $e$ is a bridge then $T(G,x,y)= xT(G.e,x,y)$,
    \item 
    If $e$ is neither loop nor bridge than
    $$
    T(G,x,y)= T(G-e,x,y)+ T(G.e,x,y).
    $$
    
\end{itemize}
\end{theorem}

Next, we introduce an equivalent expression of the Tutte polynomial as the {\em coboundary polynomial}, and the {\em bad colouring polynomial}, see \cite{B}.

\begin{definition}[coboundary polynomial]
    \label{def.cobad}
Let $G$ be a graph. The coboundary polynomial is defined by
$$
{\bar B}(G,x,y)= \sum_{A\subset E \text{flat}}p(G.A,x)y^{|A|}.
$$
\end{definition}

\begin{definition}[bad coloring polynomial]
    \label{def.bad}
    Let $G= (V,E)$ be a graph and let $x,y$ be variables.
    The {\em bad coloring polynomial} of $G$, denoted by $B(G,x,y$, is defined by the following formula:
    $$
    B(G,x,y)= \sum_{A\subset E \text{ flat}}C(G.A,x)y^{|A|}.
    $$
\end{definition}

The name is justified by the following straightforward observation.

\begin{observation}
    \label{o.bad}
Let $G$ be a graph and let $x$ be a positive integer. The bad coloring polynomial satisfies
$$
B(G,x,y)= \sum_jb_j(G,x)y^j,
$$
where $b_j(G,x)$ denotes the number of x-colourings of $G$ with exactly $j$ monochromatic edges.
\end{observation}

The following observation follows directly from the definition.

\begin{observation}
    \label{o.chcol}
    $B(G,x,y)= x^{c(E)}{\bar B}(G,x,y)$.
\end{observation}

\begin{theorem}[\cite{B}]
     \label{thm.tuba}
    The Tutte polynomial and the coboundary polynomial are related as follows:
    $$
     T(G,x,y)=(y-1)^{-r(E)}{\bar B}(G,(x-1)(y-1),y)), 
    $$
    $$
    \bar B(G,x,y)=(y-1)^{r(E)}T(G,\frac{x+y-1}{y-1},y).
    $$
\end{theorem}

\subsection{The partition function of the Potts model}
\label{sub.multit}
Next, we will introduce the {\em multivariable extension of the Tutte polynomial}, also called the partition function of the {\em Potts model} (see \cite{Sok}).
\begin{definition}
    \label{def.potts}
    Let $G= (V,E)$ be a graph and let $y= (y_e)_{e\in E}$ be the vector of variables associated with edges of $G$. Let $q$ be a positive integer. The q-state Potts model partition function for
the graph G is defined by
$$
Z(G,q,(y_e)_{e\in E})= \sum_{s:V\rightarrow [q]}\prod_{e=\{u,v\}\in E}(1+ y_e\delta(s(u),s(v)),
$$
where $\delta(a,b)= 1$ if $a=b$ and $\delta(a,b)= 0$ otherwise.
\end{definition}


Next theorem is called Fortuin-Kasteleyn representation of the Potts model (see \cite{Sok}). 

\begin{theorem}
    \label{thm.forka}
    Let $G= (V,E)$ be a graph and let $y= (y_e)_{e\in E}$ be the vector of variables associated with edges of $G$. Let $q$ be a positive integer. Then
    $$
    Z(G,q,(y_e)_{e\in E})= \sum_{A\subset E}q^{c(A)}[\prod_{e\in A}y_e].
    $$
\end{theorem}
\begin{proof}
 $$
Z(G,q,(y_e)_{e\in E})= \sum_{s:V\rightarrow [q]}\prod_{e=\{u,v\}\in E}(1+ y_e\delta(s(u),s(v))=
$$ 
$$
\sum_{s:V\rightarrow [q]}\prod_{e\in m(s)}(1+y_e),
    $$
    where $m(s)$ denotes the set of monochromatic edges of $s$. We further calculate:
    $$
  \sum_{s:V\rightarrow [q]}\prod_{e\in m(s)}(1+y_e)=
    \sum_{s:V\rightarrow [q]}\sum_{A\subset m(s)}\prod_{e\in A}y_e = \sum_{A\subset E}\sum_{s:V\rightarrow [q]; A\subset m(s)}\prod_{e\in A}y_e=
    $$
    $$
    \sum_{A\subset E}[\prod_{e\in A}y_e]q^{c(A)}.
    $$
  
  \end{proof}

  The previous theorem and Theorem \ref{thm.tch} motivate the next definition of {\em multivariate Tutte  polynomial}. We denote it by $\mathcal{T}$ to mark its difference from the Tutte polynomial (denoted by $T$).

  \begin{definition}
      \label{def.potpol}
Let $G= (V,E)$ be a graph and let $y= (y_e)_{e\in E}$ be the vector of variables associated with edges of $G$. Let $x$ be another variable. Then we define the {\em multivariate Tutte  polynomial} by
    $$
    \mathcal{T}(G,x,(y_e)_{e\in E})= \sum_{A\subset E}x^{c(A)}[\prod_{e\in A}y_e].
    $$
      
  \end{definition}

  Next, we also introduce the  multivariate bad coloring polynomial.

\begin{definition}[multivariate bad coloring polynomial]
    \label{def.bc}
    Let $G= (V,E)$ be a graph and let $y= (y_e)_{e\in E}$ be a vector of variables associated with edges of $G$. 
    The {\em multivariate bad coloring  polynomial} of $G$, denoted by $B(G,x,(y_e)_{e\in E})$, is defined by the following formula:
    $$
    B(G,x,(y_e)_{e\in E})= \sum_{A\subset E \text{ flat}}C(G.A,x)\prod_{e\in A}y_e.
    $$
\end{definition}

The multivariate Tutte polynomial and the multivariate bad coloring polynomial are related as follows:
\begin{theorem}
    \label{thm.mtuba}
    Let $G= (V,E)$ be a graph and let $y= (y_e)_{e\in E}$ be a vector of variables associated with edges of $G$. Then
    $$
B(G,x,(y_e+1)_{e\in E})= \mathcal{T}(G,x,(y_e)_{e\in E}).
    $$
\end{theorem}
\begin{proof}
$$
    B(G,x,(y_e+1)_{e\in E})= \sum_{A\subset E \text{ flat}}C(G.A,x)\prod_{e\in A}(y_e+1)=
    \sum_{A\subset E \text{flat}}[\prod_{e\in A}(y_e+1)] \sum_{B\subset E-A} (-1)^{|B|}x^{c(A\cup B)}=
    $$
    $$
    \sum_{A\subset E}[\prod_{e\in A}(y_e+1)] \sum_{B\subset E-A} (-1)^{|B|}x^{c(A\cup B)}=
     \sum_{C\subset A\subset E}[\prod_{e\in C}(y_e)] \sum_{B\subset E-A} (-1)^{|B|}x^{c(A\cup B)}=
    $$
    $$
    \sum_{C\subset E}[\prod_{e\in C}(y_e)] \sum_{C\subset A, B\subset E-A} (-1)^{|B|}x^{c(A\cup B)}=
    \sum_{A\subset E}x^{c(A)}[\prod_{e\in A}y_e].
    $$
     The third equation follows from an observation that for $A\subset E$ not flat, 
     $$
     \sum_{B\subset E-A} (-1)^{|B|}x^{c(A\cup B)}= 0.
     $$
     The  last equation follows from an observation that for each $C\subset E$,
     $$
     \sum_{C\subset A, B\subset E-A, C\neq A\cup B} (-1)^{|B|}x^{c(A\cup B)}=0.
     $$

\end{proof}

\subsection{Deletion contraction formulas}
\label{sub.delcont}
We state the deletion contraction formulas that we need to prove the main results.

\begin{theorem}[\cite{Sok}]
\label{thm.delcontrmult}
Let $G=(V,E)$ be a graph and $a\in E$.
 $$
    \mathcal{T}(G,x,(y_e:e\in E))= \mathcal{T}(G-a,x,(y_e:e\in E-a))+y_a \mathcal{T}(G.a,x,(y_e:e\in E-a)).
    $$
\end{theorem}
As an immediate consequence of Theorem \ref{thm.mtuba}, we get 
\begin{theorem}
\label{thm.delbad}
Let $G=(V,E)$ be a graph, and $a\in E$.
 $$
    B(G,x,(y_e:e\in E))= B(G-a,x,(y_e:e\in E-a))+(y_a-1) B(G.a,x,(y_e:e\in E-a)).
    $$
\end{theorem}

We will also use 
\begin{theorem}[\cite{BL}]
\label{thm.delcontrchrom}
Let $G=(V,E)$ be a graph and $a\in E$.
 $$
    C(G,x,y)= C(G-a,x,y)- C(C.a,x,y).
    $$
\end{theorem}

\subsection{Main results}
\label{sub.main}
In this section, we prove our two main results. The following Theorem \ref{thm.bad} expresses the Potential of a couple game. 

\begin{theorem}
    \label{thm.bad}
Let $G=(V,E)$ be a graph. Let $w(G)$ be a couple game defined by $w:2^E\rightarrow \mathcal Q^+$, i.e., for each $\emptyset\neq A\subset E$, $w(A)= \prod_{e\in A}w(e)$ and $w(\emptyset)= 0$
(see Definition \ref{def.coup}). 
The potential $P(w(G))$ satisfies:
 $$
 P(w(G))= \sum_{T\subset E \text{ flat}}w(T)C(G.T,x)= B(G,x,(w(e))_{e\in E})- C(G,x).
$$  
\end{theorem}
\begin{proof}
 It follows from Theorem \ref{thm.chrom}, Definition \ref{def.bc} that
 $$
 P(w(G))= \sum_{T\subset E \text{ flat}}w(T)C(G.T,x)= 
 $$
 $$
 \sum_{\emptyset\neq T \text{ flat}}[\prod_{e\in T}w(e)] C(G.T,x)= B(G,x,(w(e))_{e\in E})- C(G,x).
 $$
\end{proof}

Next, we turn our attention to the Shapley value. We recall (see Definition  \ref{def.pot}) that the potential $P(w)$ and the Shapley values $\phi_e(w)$ of a cooperative game $w$ satisfy $\phi_e(w)= P(w)- P(w|({E-e}))$. This suggests that a deletion-contraction formula can be used to determine the Shapley values of couple games using previous Theorem \ref{thm.bad}.

\begin{theorem}
    \label{thm.dcc}
   Let $G=(V,E)$ be a graph and let $w(G)$ be a couple game as in Theorem \ref{thm.bad}. Let $a\in E$ be an edge of $G$.
The Shapley value  $ \phi_a(w(G))$ of the couple game $w(G)$ satisfies: 
$$
\phi_a(w(G))= (w(a)-1)B(G.a,x,(w_e)_{e\in E-a})+ C(G.a,x).
$$
\end{theorem}
\begin{proof}
We have, by Lemma \ref{l.shapleygraphic} and Lemma \ref{l.potgraphic}, that 
$$
 P(w(G))- \phi_a(w(G))= \sum_{\emptyset\neq R\subset E-a}x^{c(R)}c_R(w)= P(w(G-a))=
 B(G-a,x,(w(e))_{e\in E-a})- C(G-a,x).
$$
The last equation is Theorem \ref{thm.bad}. Now, we use Theorem \ref{thm.mtuba} and the deletion contraction formulas of Theorem \ref{thm.delbad} and Theorem \ref{thm.delcontrchrom}.
    
    \end{proof}

In particular, 
\begin{theorem}
    \label{thm.dcc1}
   Let $G=(V,E)$ be a graph and let $w(G)$ be a couple game such that $w(A)=1$ for each $\emptyset\neq A\subset E$. Let $a\in E$ be an edge of $G$. Then
$$
\phi_a(w(G))= C(G.a,x).
$$
\end{theorem}

\section{An Extension}
\label{s.qchrom}

In this section, we generalize the connectivity augmentation process introduced in Section \ref{s.connectivity}: we replace the simple exponential by a more complex function which models a wider range of situations. We show that such a generalization still leads to games whose potential and Shapley values are related to studied chromatic functions.

Let $k>0$ be a natural number. The {\em quantum integer} is defined as $(k)_q= q^{k-1}+\ldots + q+ 1= \frac{q^k-1}{q-1}$. We also let $0_q= 0$. We note that $(k)_q|_{q=1}= k$.

We proceed consistently as in Section \ref{s.connectivity}. We first modify Definition \ref{def.graph-gamebasis}.

If $u:V\rightarrow \mathcal Q$ then for each $U\subset V$, we let $u(U)= \sum_{v\in U} u(v)$.

\begin{definition}[Basic q-connectivity augmented games]    \label{def.qgraph-gamebasis}
    Let $G=(V,E)$ be a graph, $u:V\rightarrow \mathcal Q$, and $x$ be a number. For $R\subset E$, let $b_q(R,G,u,x)$ be the characteristic function (of a cooperative game) defined as follows:
if $R\subset S\subset E$ then 
$$
b_q(R,G,u,x)(S)= |R|\prod_{W \text {component of }(V,R)}(x)_{q^{u(W)}},
$$ 
otherwise $b_q(R,G,u,x)(S)= 0$.
We call $b_q(R,G,u,x)$ a {\em basic q-connectivity augmented game}.
\end{definition}

We note that $b_q(R,G,u,x)|_{q=1}= b(R,G,x)$.

\begin{definition}[q-connectivity augmented games]
    \label{def.qconaugame}
 Let $G=(V,E)$ be a graph, $u:V\rightarrow \mathcal Q$, and $x$ be a number. Let $w:2^E\rightarrow \mathcal Q^+$ be a cooperative game of $\mathcal E$. We denote by $w_q(G,u)$ the cooperative game $w_q(G,u)=\sum_{R\subset E}c_R(w) b_q(R,G,u,x)$. We call $w_q(G,u)$ {\em q-connectivity augmented $w$}, or a {\em q-connectivity augmented game}.
\end{definition}

\begin{definition}
    \label{def.qcoup}
     Each q-connectivity augmented game $w_q(G,u)$ such that $w\in \mathcal E$ satisfies, for each $\emptyset\neq A$ that $w(A)= \prod_{e\in A}w(e)$,
      is called {\em q-couple game} (see Definition \ref{def.coup} of the couple game). 
    \end{definition}

Next lemma is analogous to Lemma \ref{l.potgraphic}.

   \begin{lemma}
 \label{l.qpotentialgraphic}   
Let $G=(V,E)$ be a graph and $u:V\rightarrow \mathcal Q$. The {\em Potential} $P(w_q(G,u))$ of the q-connectivity augmented game $w_q(G,u)$ satisfies:
$$
P(w_q(G,u))= 
\sum_{T\subset E}w(T) \sum_{T\subset R}(-1)^{|R|-|T|}
\prod_{W \text {component of }(V,R)}(x)_{q^{u(W)}}.
$$
\end{lemma}

\subsection{q-Chromatic function}
\label{sub.qchrom}

Let $k$ be a positive integer. Following \cite{L}, we let  
$$
C_q(G,x)= \sum_{s\in V(G,k)}q^{\sum_{v\in V}s(v)},
$$
where
$V(G,k)= \{s:V\rightarrow \{0, \dots, k-1\}; s \text{ proper coloring of }G\}$.

\begin{theorem}[\cite{L}]
    \label{thm.wh}
 Let $G=(V,E)$ be a graph, $x$ be a positive integer. Then
 $$
 C_q(G,x)= \sum_{A\subset E}(-1)^{|A|}\prod_{W \text {component of }(V,A)}(x)_{q^{|W|}}.
 $$
 \end{theorem}

\begin{definition}[q-chromatic function \cite{L}]
\label{def.qch}
Let $G=(V,E)$ be a graph and let $x$ be a variable. The q-chromatic function of $G$ is 
 $$
 C_q(G,x)= \sum_{A\subset E}(-1)^{|A|}\prod_{W \text {component of }(V,A)}(x)_{q^{|W|}}.
 $$
\end{definition} 

Next, we define {\em weighted q-chromatic function}. If $u:V\rightarrow \mathcal Q$ then for each $U\subset V$, we let $u(U)= \sum_{v\in U} u(v)$.

\begin{definition}[weighted q-chromatic function]
\label{def.wqch}
Let $G=(V,E)$ be a graph and let $x$ be a variable. Let $u:V\rightarrow \mathcal Q$ be a function giving a weight to each vertex. The weighted q-chromatic function of $G$ is 
 $$
 C_q(G,u,x)= \sum_{A\subset E}(-1)^{|A|}\prod_{W \text {component of }(V,A)}(x)_{q^{u(W)}}.
 $$
\end{definition} 

We now extend the definition of {\em contraction} of a subset of edges to graphs with weighted vertices.

\begin{definition}[contraction in vertex-weighted graphs]
  \label{def.qcontr} 
  Let $G=(V,E)$ be a graph, let $u:V\rightarrow \mathcal Q$ be a function that gives a weight to each vertex, and let $A\subset E$. We denote by 
  $(G,u).A$ the pair $(G.A, u.A)$, where $G.A$ is the standard contracted graph defined in Definition \ref{def.contr}, 
  and $u.A: \mathcal W\rightarrow \mathcal Q$ assigns weights to the vertices of $G.A$ (equivalently, to the components of $(V,A)$), whose set is denoted by $\mathcal W$. The contracted weights satisfy: 
For each $w\in \mathcal W$, $u.A(w)= \sum_{v\in w} u(v)$.
\end{definition}

Now, we are ready to extend Theorem \ref{thm.chrom}.

\begin{theorem}
\label{thm.qchrom}
Let $G=(V,E)$ be a graph, $u:V\rightarrow \mathcal Q$ vertex-weights, let $x$ be a variable, and let $w$ be a cooperative game of $\mathcal E$. The Potential $P(w_q(G,u))$ of the q-augmented connectivity game $w_q(G,u)$ satisfies:
$$ 
P(w_q(G,u)) = \sum _ {R\subset E \text{ flat}} w(R)C_q((G,u).R,u.R,x),
$$
where $(G,u).R$ is the graph obtained from $(G,u)$ by contraction of edges of $R$. For $C_q$ see Definition \ref{def.qch}.
\end{theorem}

\begin{proof}
By Lemma \ref{l.qpotentialgraphic},
$$
P(w_q(G,u))= 
\sum_{T\subset E}w(T) \sum_{T\subset R}(-1)^{|R|-|T|}
\prod_{W \text {component of }(V,R)}(x)_{q^{u(W)}}=
$$
   $$
\sum_{T\subset E}w(T) \sum_{Q\subset E\setminus T}(-1)^{|Q|}
\prod_{W \text {component of }(V,Q\cup T)}(x)_{q^{u(W)}}=
$$ 
 $$
\sum_{T\subset E \text{ flat}}w(T) \sum_{Q\subset E\setminus T}(-1)^{|Q|}
\prod_{W \text {component of }(V((G,u).T),Q)}(x)_{q^{u(W)}}=
$$ 
$$
\sum _ {T\subset E} w(T)C_q(G.T,u.T,x).
$$ 
\end{proof}

Results of Section \ref{s.tutte} can also be extended to q-couple games.

\section{Conclusion}
\label{s.concl}
We initiate the study of cooperative games where there exist {\em groups of pre-aligned agents}. 
In this paper, each pre-aligned group has size two. 
Our goal is to find a way to determine the contribution of each individual local connection (pre-aligned couple) to the connectivity of the (reliability) network, where possibly each edge is given with the probability of not-failing. We model these contributions by Shapley values of the connectivity augmented 'local values' which are characteristic functions whose values are determined, for example, by the probabilities of not-failing.
We link such Shapley values to the world of chromatic and Tutte polynomials and the partition function of the Potts model from statistical physics.
Future work is to broaden the described basic concept to a more general setting.

\bibliographystyle{alpha}
\bibliography{main}

\end{document}